\documentclass[,,]{eptcs}

\usepackage[pdftex]{graphicx}
\usepackage{latexsym, amsgen, amsmath, amstext, amsbsy, amsopn, amsfonts, amsthm, amssymb}
\newtheorem{lemma}{Lemma}
\newtheorem{theorem}{Theorem}

\newcommand{\tnMarkZero}{\ensuremath{0\hspace{-1.2ex}/\,}}
\newcommand{\tnMarkOne}{\ensuremath{1\hspace{-1.2ex}/\,}}

\newcommand{\bdiv}{\mathop{\mathrm{div}}}
\newcommand{\via}{\mathop{\mathrm{via}}}

\providecommand{\urlalt}[2]{\href{#1}{#2}}
\providecommand{\doi}[1]{doi:\urlalt{http://dx.doi.org/#1}{#1}}

\title{Small Polynomial Time Universal Petri Nets}
\author{Dmitry A. Zaitsev
\institute{International Humanitarian University\\
Department of Computer Engineering}
\institute{Fontanskaya Doroga, 33, Odessa 65009, Ukraine}
\email{daze@acm.org}
}

\def\titlerunning{Polynomial Time Small Universal Petri Nets}
\def\authorrunning{D.A. Zaitsev}
\begin{document}
\maketitle

\begin{abstract}
The time complexity of the presented in 2013 by the author small universal Petri nets with the pairs of places/transitions numbers (14,42) and (14,29) was estimated as exponential. In the present paper, it is shown, that their slight modification and interpretation as timed Petri nets with multichannel transitions, introduced by the author in 1991, allows obtaining polynomial time complexity. The modification concerns using only inhibitor arcs to control transitions' firing in multiple instances and employing an inverse control flow represented by moving zero. Thus, small universal Petri nets are efficient that justifies their application as models of high performance computations. 
\end{abstract}

\section{Introduction}

The first explicitly constructed universal Petri net (UPN) \cite{Zv12UPN} consists of thousands of vertices and works exponentially slow; for its design, a direct simulation of Petri net behavior according to its state equation was applied. The construction was done within the class of pure inhibitor Petri nets.

Small universal Petri nets with the pairs of places/transitions numbers (14,42) \cite{Zv13-14-29} and (14,29) \cite{Zv13-14-42} were constructed using the technique of simulating know small universal Turing machines (UTM): the net UPN(14,42) simulates UTM(6,4) \cite{Ne08,NW09FI} and the net UPN(14,29) simulates weakly universal Turing machine (2,4) -- WUTM(2,4) \cite{Ne08,NW09W}. The construction was done within the class of deterministic inhibitor Petri nets. Their time complexity was estimated as exponential. 

In the present work, these nets are modified: a place LEFT added to represent the left moves of the TM control head in an explicit way; only inhibitor arcs are used to control transitions which fire in multiple instances; the control flow is represented in an inverse way with moving zero to provide check via inhibitor arcs. Interpretation of the modified UPN as a Petri net with multichannel transitions \cite{Zv91,ZvSl97} gives polynomial estimations of time complexity. Thus, constructed UPN are efficient models of computations.

In sequel, net PolyUPN(15,29) is constructed on UPN(14,29), and it is shown that interpreting its behavior as a deterministic timed inhibitor Petri net which transitions' firing time is discrete and equal to unit and unlimited number of transitions' channels gives polynomial time complexity; the corresponding class of Petri net is named a deterministic arithmetic Petri net. In the same way the results could be obtained for other early constructed universal Petri nets.

\section{A deterministic arithmetic Petri net}

A class of \textit{deterministic arithmetic Petri net}  is employed in the present paper that represents a combination of \textit{a deterministic inhibitor Petri net} (DIPN) \cite{Zv13-14-29,Zv13-14-42} and \textit{a timed Petri net with multichannel transitions} \cite{Zv91,ZvSl97} for the case when the firing time of each transition is equal to unit. Since all the times are equal to unit, it is more convenient to not introduce the time concept at all but to perceive Petri net as a synchronous \cite{Pe81} net when at each step the maximal valid set of firable transitions fires. The validity of the maximal set is defined as not leading to negative markings. But when using transitions' numbers as priorities in deterministic inhibitor Petri net, only one transition could fire at a step possibly in a few instances. The necessity of restricting the net behavior making it deterministic is stipulated by the deterministic character of simulated TM.

\textit{A deterministic arithmetic Petri net} (DAPN) is is a bipartite directed multigraph supplied with a dynamic process that has synchronous deterministic behavior. DAPN is denoted as a quadruple $N=(P,T,F,{\mu}^0)$ , where $P$ and $T$ are disjoint sets of vertices called places and transitions respectively, the mapping $F$ defines arcs between vertices, their type and multiplicity,  and the mapping ${\mu}^0$ represents the initial state (marking). The transition choice order is defined by their enumeration $T=\{t_1,t_2,...,t_n\}, n=|T|$; places are also supposed been enumerated $P=\{p_1,p_2,...,p_m\}, m=|P|$, and the places' marking is represented as a vector $\overline{\mu}, {\mu}_j={\mu}(p_j)$ with integer index $j, 1 \le j \le m$, where ${\mu}_j$ is a nonnegative integer equal to the number of tokens situated in place $p_j$. 

The mapping ${F:(P \times T)\rightarrow \{0,w^{-},-1\} \cup (T\times P)\rightarrow\{0,w^{+}\}}$ defines arcs, their types and multiplicities, where transition input and output arc multiplicities $w^{-}$ and $w^{+}$ are natural numbers, a zero value corresponds to the arc absence, a positive value -- to the regular arc with indicated multiplicity, and a minus unit -- to the inhibitor arc. As it was shown in \cite{Pe81}, Petri nets with the multiple arcs are easily converted to \textit{ordinary} Petri nets with regular arcs' multiplicity equal to unit. To avoid nested indices we denote $w^{-}(p_j, t_i)$ as $w^{-}_{j,i}$ and $w^{+}(t_i, p_k)$ as $w^{+}_{i,k}$.

In graphical form, places are drawn as circles and transitions as rectangles. An inhibitor arc is represented by a small hollow circle at its end, and a small solid circle (read arc) represents the abbreviation of a loop. Regular arc's multiplicity greater than unit is inscribed on it and place's marking greater than zero is written inside it.

To estimate firability conditions on each incoming arc of a transition, the following auxiliary operation is defined
$$x \via y=\left\{
\begin{array}{lll}
x \bdiv y & if~y>0, \\
0 & if~y=-1,~x>0, \\
\infty & if~y=-1,~x=0.
\end{array}
\right.$$
To avoid inconsistency with infinite number of instances, here we prohibit transitions without input regular arcs.

The \textit{behavior} (dynamics) of a DAPN could be described by the corresponding state equation similarly to \cite{Zv12UPN,Zv91,ZvSl97}. The present work considers the behavior as result of sequential applying the following \textit{transition firing rule}: 

\begin{enumerate}

\item net transitions ${t_{i}}$ are checked sequentially with integer index $i$ ranging from $1$ to $n$;

\item the number of instances of transition ${t_{{i}}}$ firable at the current step is equal to

$v_i=v(t_{i})=\min\limits_{j}{({\mu}_j \via w^{-}_{j,i})},~1 \le j \le m,~F(p_{{j}},t_{{i}}) \neq 0$.  

\item the first transition ${t_{i}}$ (with the minimal value of index ${i}$) having $v_i>0$ fires in $v_i$ instances; 

\item when transition ${t_{{i}}}$ fires, it

\begin{enumerate}
\item extracts ${v_i}*w^{-}_{j,i}$ tokens from each its input place $p_{{j}}$ for regular arcs $F(p_j,t_i)>0$;

\item puts ${v_i}*w^{+}_{i,k}$ tokens into each its output place $p_{{k}}$, $F(t_{{i}},p_{{k}})>0$;
\end{enumerate}

\item the net halts if firable transitions are absent.

\end{enumerate}

When a transition, having a single regular incoming arc with multiplicity $x$ from place $p$ and a single regular outgoing arc with multiplicity $y$ to  place $p'$, fires, it implements the following computations $\mu(p):=\mu(p) \bmod x,~ \mu(p'):=\mu(p') + y \cdot (\mu(p) \bdiv x)$. That is why it is named arithmetic Petri net: it implements division by $x$ and multiplication by $y$. Choosing either $x$ or $y$ equal to unit we obtain either pure multiplication or pure division.

Thus, transitions could be thought of as virtual actions. The number of really started actions depends on the amount of available resources represented by transitions' input places. Why we should restrict the number of transitions' instances to unit and fire them in sequence, as in classical Petri net, when available resources allow firing them simultaneously? Anyway, classical sequential order of transitions' firing could be obtained as a special case attaching a place to each transition connected with read arc and having marking equal to unit.

\section{Constructing PolyUPN(15,29)}

As in UPN(14,29) \cite{Zv13-14-29}, the source information for simulation is the transition function of WUTM(2,4) \cite{Ne08,NW09W} and the encoding of its states and tape symbols given by Table~\ref{tab:wutm24}. In WUTM(2,4), an infinite repetition of definite blank words is written on its tape: ${w_{{l}}=00{\tnMarkZero}1}$ to the left and ${w_{{r}}=0{\tnMarkOne}{\tnMarkZero}{\tnMarkZero}0{\tnMarkOne}}$ to the right of the working zone.  According to the function $s(x_{{l-1}}x_{{l-2}}{...}x_{{0}})=\sum_{i=0}^{l-1}{s(x_{{i}})\cdot 5^{{i}}}$ for the tape words' encoding \cite{Zv13-14-29}, the codes of the left and right blank words are: ${{sw}_{{l}}={167}}$ and ${{sw}_{{r}}={13596}}$.

\begin{table}
  \begin{center}
    \begin{tabular}{|c|c|c|c|}
      \hline
      $\Sigma \backslash \Omega$ &  & $u_1$ & $u_2$\\
      \hline
        & $s(\Sigma) \backslash s(\Omega)$ & 0 & 1\\
      \hline
      0 & 1 & $3,left,0$ & $4,right,0$\\
      \hline
      1 & 2 & $4,left,1$ & $3,left,1$\\
      \hline
      $\tnMarkZero$ & 3 & $4,left,0$ & $1,right,1$\\
      \hline
      $\tnMarkOne$ & 4 & $4,left,0$ & $2,right,1$\\
      \hline
    \end{tabular}
  \end{center}
  \caption{WUTM(2,4) behavior \cite{Ne08,NW09W} and its encoding}
  \label{tab:wutm24}
\end{table}

To obtain PolyUPN(15,29), net UPN(14,29) \cite{Zv13-14-29} was modified in the following way:

\begin{itemize}
\item place $LEFT$ ($p_{15}$) added, to represent the left moves of the control head, together with arcs from transitions simulating instructions with the left move: $t_4$, $t_6$, $t_7$, $t_8$, $t_{10}$; thus a pair of places $LEFT$ ($p_{15}$) and $RIGHT$ ($p_7$) are complimentary: only one of them contains a token when simulating a TM step;
\item a regular arc is added from place $LEFT$ ($p_{15}$) to transition $t_{29}$ to clean it after simulating the left move;
\item the control flow within the sequence of subnets $MA5LR$, $MD5LR$ represented by the sequence of vertices $p_6$, $t_{13}$, $p_{10}$, $t_{16}$, $p_{11}$, $t_{19}$, $p_8$, $t_{22}$, $p_{12}$, $t_{25}$, $p_{13}$, $t_{28}$ or $t_{29}$ was modified; a token was put into each place; reverse order of arcs was used; inhibitor arcs from the previous places were added; thus, the control flow is represented by moving zero marking;
\item all the checks of control flow places on unit with read arks are replaced by checks on zero with inhibitor arks;
\item all the checks of place $RIGHT$ ($p_7$) on unit with read arks are replaced by checks of place $LEFT$ ($p_{15}$) on zero;
\item all the arcs from transitions, simulating TM instructions $t_3 - t_{10}$ to place $MOVE$, were reveresed in their direction.
\end{itemize}

The general scheme of PolyUPN(15,29) is shown in Fig.~\ref{fig:top15-29}. Subnets are depicted as rectangles with double line border. Some vertices have mnemonic names besides their numbers. Used subnets $FS$, $MA5LR$, and $MD5LR$ are represented in Fig.~\ref{fig:fs}, Fig.~\ref{fig:ma5lr}, Fig.~\ref{fig:md5lr} correspondingly. Places with the same name (number) are considered as the same place and should be merged all over the components. The obtained final assembly of PolyUPN(15,29) is shown in Fig.~\ref{fig:polyupn15-29}. Since it is rather tangled, it could be represented in a tabular form similar to UPN(14,29) \cite{Zv13-14-29}.

Place $U$ contains encoded TM state $s(u)$, place $X$ contains encoded current cell symbol $s(x)$, and places $L$ and $R$ contain encoded left and right parts of the tape working zone respectively regarding the current cell.  At the beginning of each computation step, place $STEP$ launches subnet $FS$, which simulates  WUTM(2,4) transition function. Subnet $FS$ produces the encoding $s(u')$ of the new state and the encoding $s(x')$ of the new symbol in places $U$ and $X$ respectively to simulate the Turing machine instruction. Subnet $FS$ also puts a token into either place $RIGHT$ or place $LEFT$ to indicate the control head right or left moves correspondingly. A token is extracted from place MOVE after subnet $FS$ has finished that launches the sequence of subnets $MA5LR$, $MD5LR$, which simulates the control head moves. At the end of a simulated computation step, a token is put into place $STEP$ that allows the simulation of the next instruction to begin. Moreover, places LEFT and RIGHT are cleared. 

\begin{figure}
  \centering
    \includegraphics [width=0.6\textwidth] {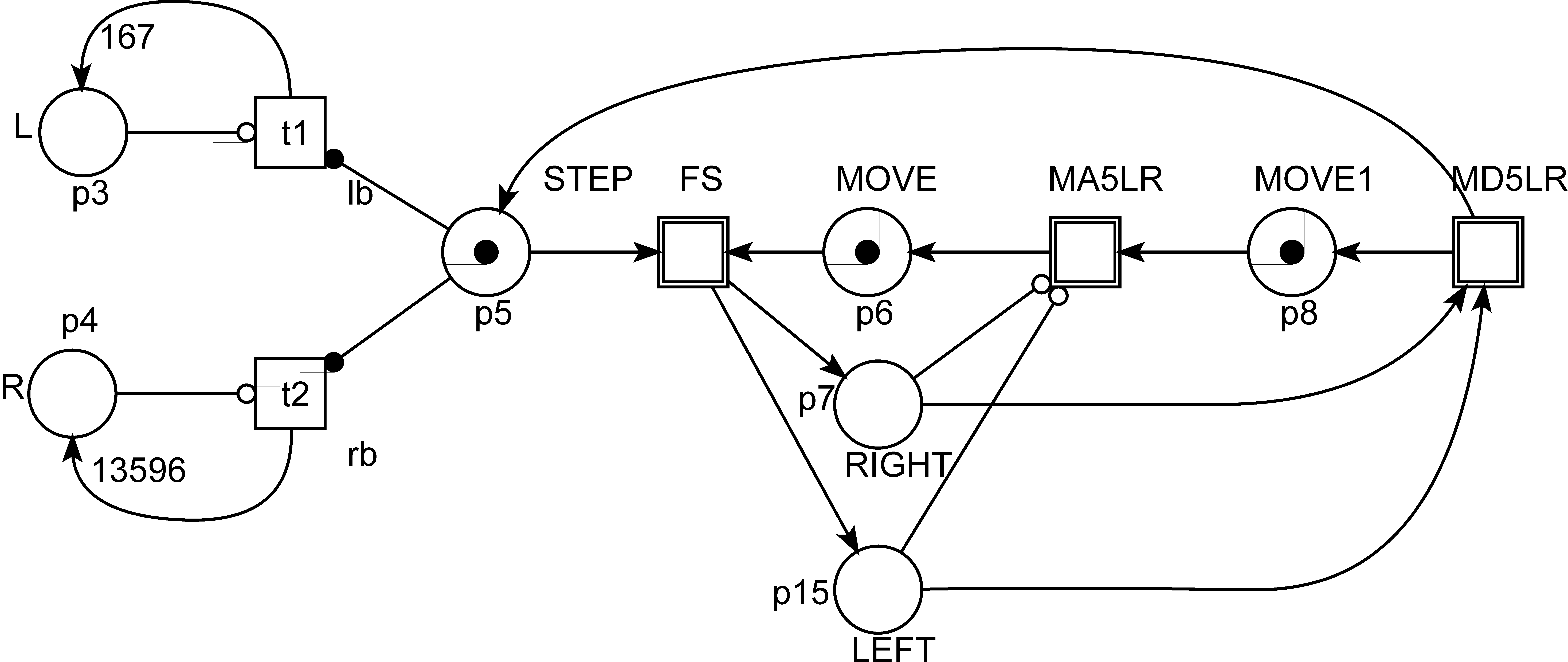}
  \caption{General arrangement of PolyUPN(15,29)}
  \label{fig:top15-29}
\end{figure}

Subnet $FS$ (shown in  Fig.~\ref{fig:fs}) simulates the transition function of WUTM(2,4) as follows: the Turing machine instruction for each pair $(x,u)$ is encoded by a transition with the label $(s(x), s(u))$ in the subnet $FS$; its input arcs from places $X$ and $U$ have corresponding multiplicity, zero multiplicity means the arc absence. For a TM instruction ${(x,u,x',v,u')}$, output arcs to places $X$ and $U$ have multiplicity ${s(x')}$ and ${s(u')}$ correspondingly with an extra arc either to place $RIGHT$ when ${v=right}$ or place $LEFT$ when ${v=left}$. Subnet $FS$ of PolyUPN(15,29) is represented in Fig.~\ref{fig:fs}; each transition has an incoming arc from place $STEP$ and place $MOVE$. The required relations of priorities are shown via dashed auxiliary arcs connecting transitions, providing only one firable transition on each step of UPN work.

\begin{figure}
  \centering
    \includegraphics [width=0.5\textwidth] {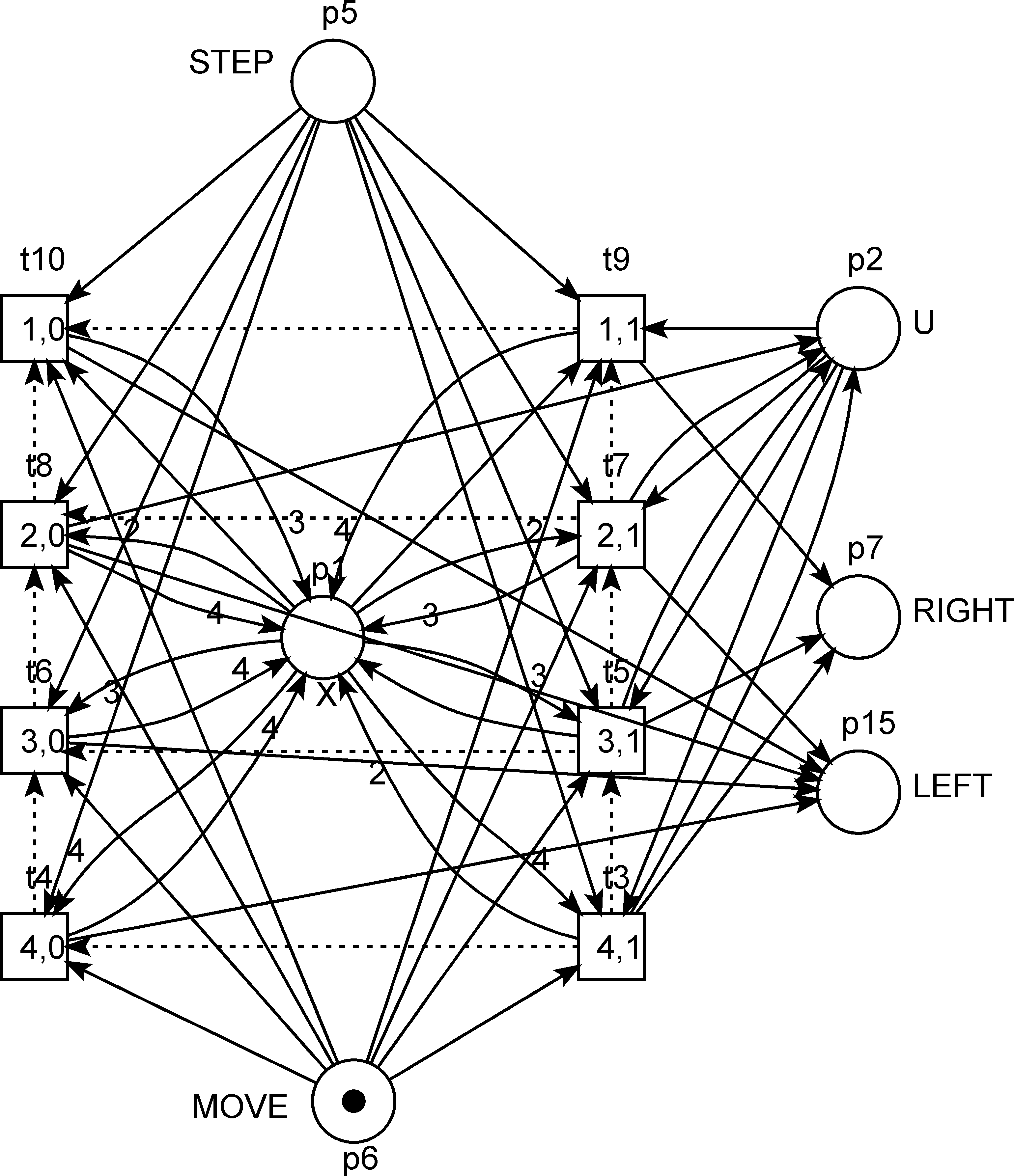}
  \caption{Subnet $FS$ simulating WUTM(2,4) transition function}
  \label{fig:fs}
\end{figure}

The TM tape is represented by the places $L$, $X$, and $R$ containing the encoded left part of the working zone, the current cell symbol and the encoded right part of the working zone correspondingly. The moves of the tape head on the tape are simulated by the two connected subnets $MA5LR$ and $MD5LR$ shown in Fig.~\ref{fig:ma5lr}, Fig.~\ref{fig:md5lr}. The meaning of subnets' names is the following: $MA5$ multiplication and addition with radix 5 (${S:=S\cdot 5+X}$), $MD5$ modulo and division with radix 5 (${S:=S \bdiv 5}$, ${X:=S \bmod 5 }$); $LR$ choice of places either $L$ or $R$, where codes of the left and right parts of the tape, regarding the current cell symbol code $X$, are stored, depending on the marking of places $LEFT$ and $RIGHT$.  Two transitions $lb$ and $rb$ simulate peculiarities of weakly universal TM work, they add the blank word codes ${{sw}_{{l}}}$ and ${{sw}_{{r}}}$ to the codes $L$ and $R$ of the left and right parts of the tape working zone correspondingly when its value is equal to zero.

\begin{figure}
  \centering
    \includegraphics [width=0.5\textwidth] {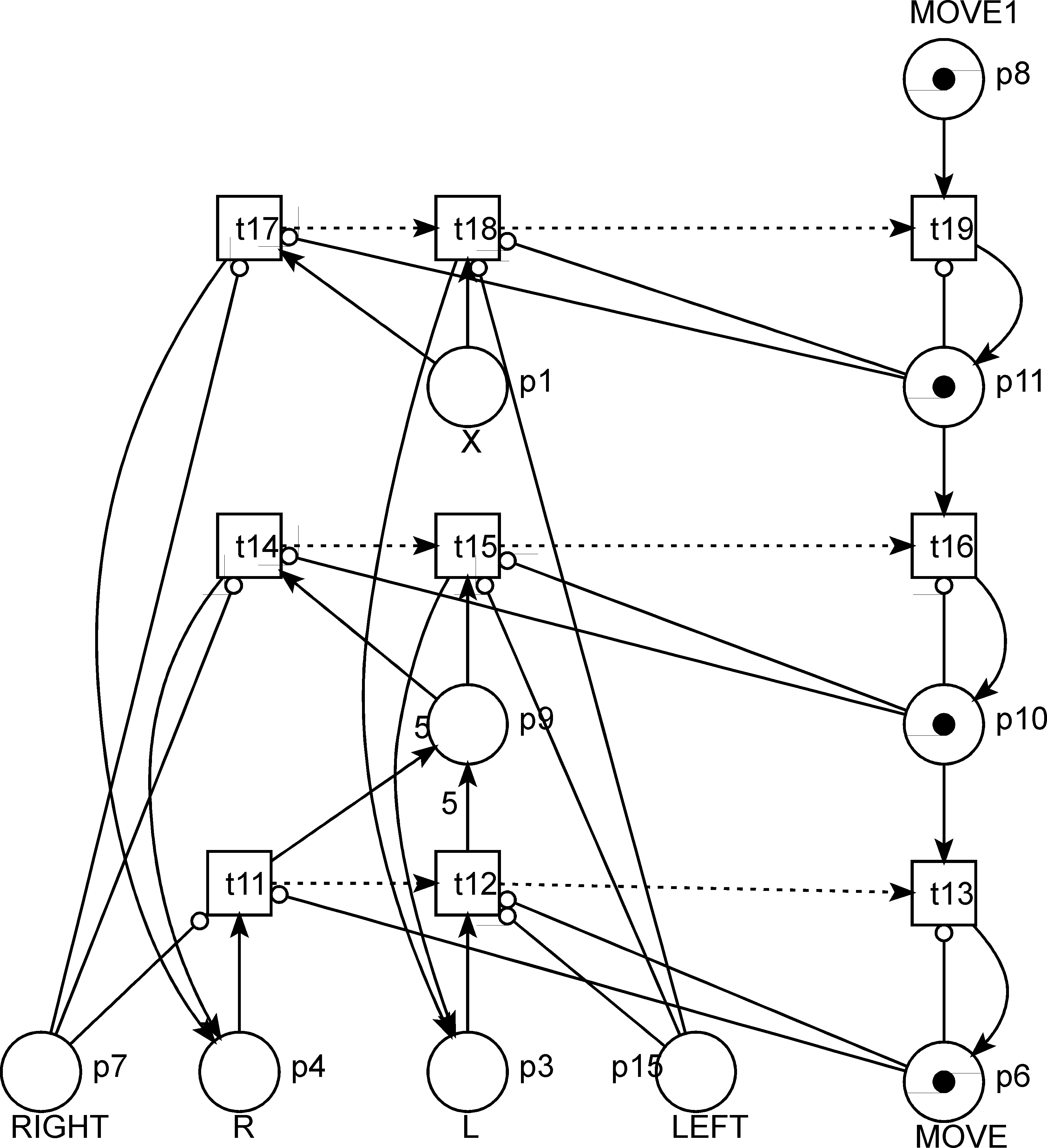}
  \caption{Basic subnet of the tape encoding $MA5LR$ -- add a symbol to the code}
  \label{fig:ma5lr}
\end{figure}

Thus, the sequence of subnets $MA5LR$, $MD5LR$ implements the following operations: 
\begin{itemize}
\item to simulate a left move (when place ${RIGHT=0}$): ${R:=R\cdot 5+X}$, ${L:=L \bdiv 5}$, ${X:=L \bmod 5}$;
\item to simulate a right move (when place ${LEFT=0}$): ${L:=L\cdot 5+X}$, ${R:=R \bdiv 5}$, ${X:=R \bmod 5}$.
\end{itemize}

\begin{figure}
  \centering
    \includegraphics [width=0.6\textwidth] {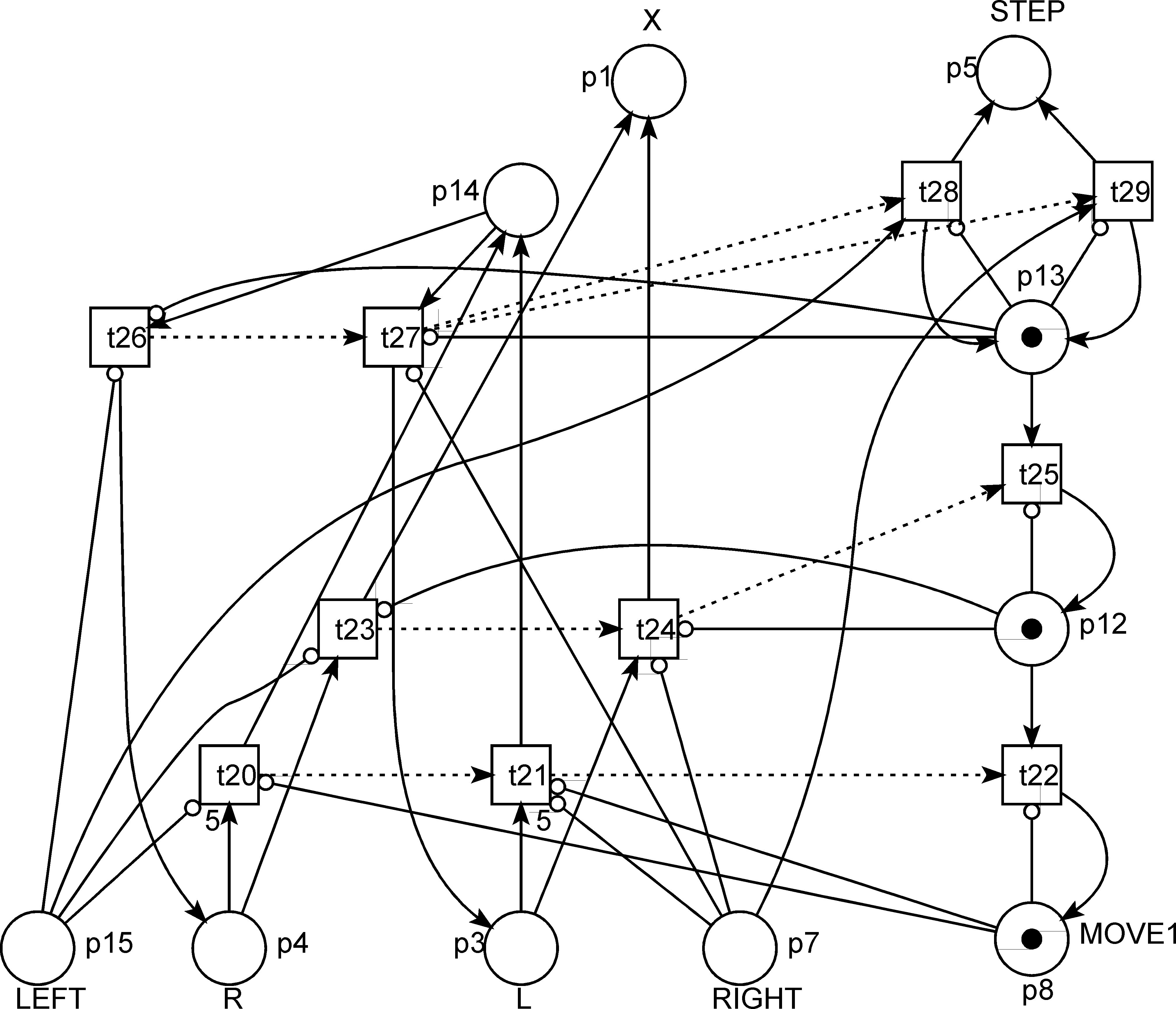}
  \caption{Basic subnet of the tape decoding $MD5LR$ -- extract a symbol from the code}
  \label{fig:md5lr}
\end{figure}

PolyUPN(15,29) is composed according to Fig.~\ref{fig:top15-29} via inserting subnets and merging places with the same names. Transitions are enumerated to provide required relations of priorities; places are enumerated in an arbitrary order; the number of arcs is 125. The obtained PolyUPN(15,29) is shown in Fig.~\ref{fig:polyupn15-29}. 

\begin{figure}
  \centering
    \includegraphics [width=0.9\textwidth] {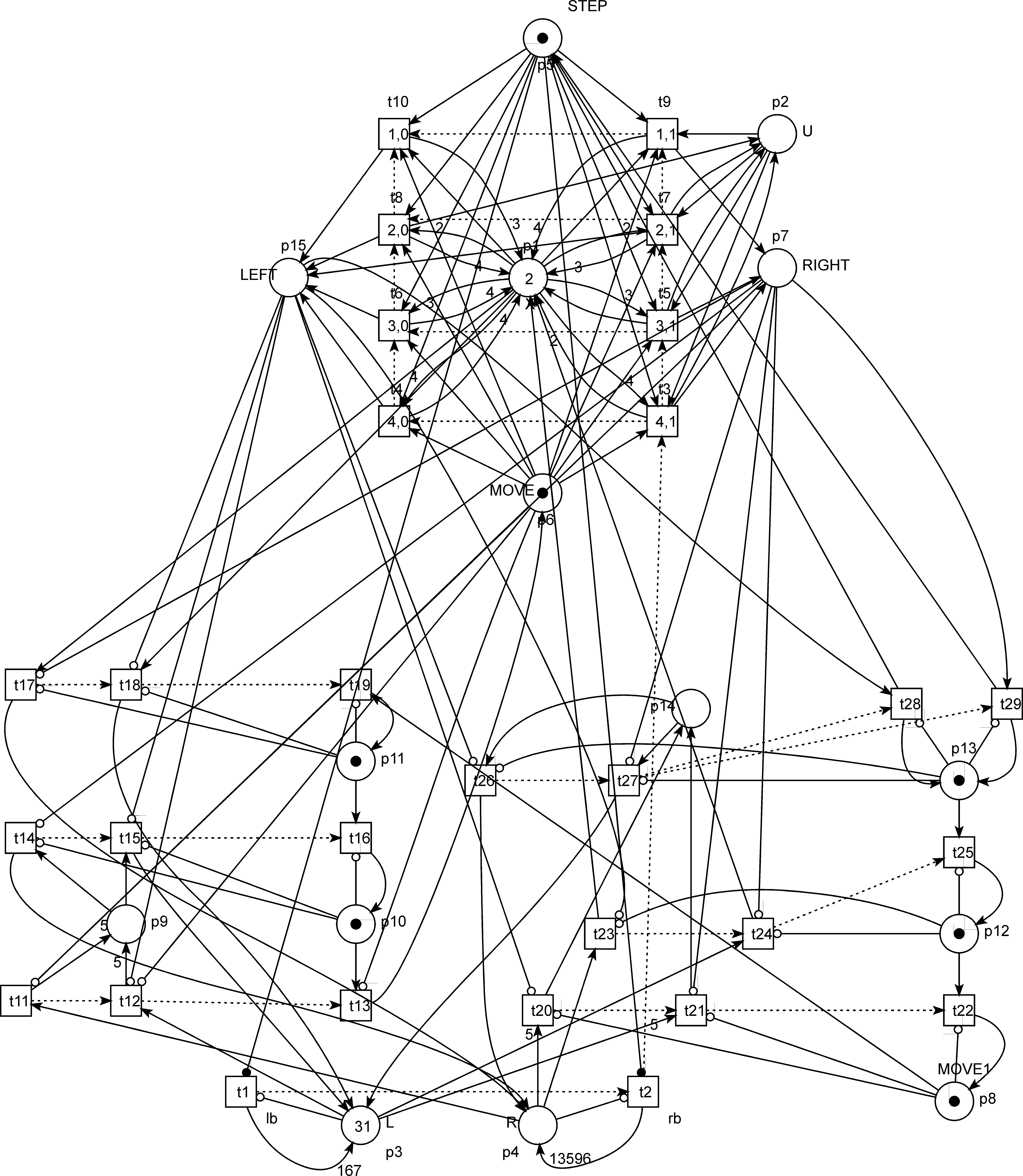}
  \caption{UPN(14,29) in graphical form}
  \label{fig:polyupn15-29}
\end{figure}

The work of PolyUPN(15,29) was simulated in the environment of system Tina (http://laas.fr/tina) in the step-by-step order; the results coincide with  the tracing table of \cite{Zv13-14-29}. Tina does not provide firing transitions in a few instances but the conditions for such firing were preserved: the only firable transition has been firing sequentially during amount of steps equal to the number of its instances in DAPN.

\section{Estimations of PolyUPN(15,29) time complexity}

\begin{lemma}
Subnet $FS$ simulates TM transition function in a single DAPN step.
\label{lem:tmfs}
\end{lemma}

The proof of the lemma is an immediate conclusion of subnet $FS$ arrangement rules \cite{Zv13-14-29,Zv13-14-42} and late modifications regarding adding place $LEFT$, reversing arcs of place $MOVE$, and initial marking of place $MOVE$ equal to unit. When $STEP=0$, none of $FS$ transitions is firable; when $STEP=1$, only one transition is firable that fires removing a token from places $STEP$ and $MOVE$ and changing marking of places $U$ and $X$ according to TM transition function. 

\begin{lemma}
The sequence of subnets $MA5LR$, $MD5LR$, supplied with transitions $lb$, $rb$, simulates work with weakly universal TM tape in not more that 13 DAPN steps.
\label{lem:tmtape}
\end{lemma}
\begin{proof}
Before simulating current TM step, only one of places $L$, $R$ can become zero: $L$ -- for the left move and $R$ -- for the right move as a result of division operation that means hitting the corresponding edge of the tape working zone. Then place $STEP$ enables one of transition $lb$, $rb$ which fires before $FS$ work and inserts the corresponding blank word code disabling the fired transition. Then, as $lb$, $rb$ are disabled, subnet $FS$ starts. 

After $FS$ finished, marking of place $MOVE$ is zero, marking of places $U$ and $X$ corresponds to the new state/symbol pair and place $LEFT$ is marked for instructions with the left move while place $RIGHT$ is marked for instructions with the right move. Let us consider an instruction with the left move when $LEFT=1,RIGHT=0$. 

Subnet $MA5LR$ works first. The only firable transition $t_{11}$ fires in $R$ instances putting $R \cdot 5$ tokens into place $p_9$. Then $t_{13}$ fires taking a token from $p_{10}$ and putting a token to $p_6$ ($MOVE$) that looks like moving zero from $p_6$ ($MOVE$) to $p_{10}$. The only firable transition $t_{14}$ fires in $R \cdot 5$ instances putting $R \cdot 5$ tokens into place $R$ ($R:=R \cdot 5$). Then $t_{16}$ fires moving zero from $p_{10}$ to $p_{11}$. The only firable transition $t_{17}$ fires in $X$ instances adding $X$ tokens into place $R$ and cleaning place $X$. Then $t_{19}$ fires moving zero from $p_{11}$ to $p_8$ ($MOVE1$) that disables transitions of subnet $MA5LR$ and enables transitions of subnet $MD5LR$. The result of subnet $MA5LR$ work is the following $X:=0,R:=R \cdot 5 +X$ and marking of internal place $p_8$ is equal to zero that is achieve via the only transitions firing sequence ${t_{11}^{R}}{t_{13}}{t_{14}^{R \cdot 5}}{t_{16}}{t_{17}^{X}}{t_{19}}$. As the control flow places $p_6$, $p_{10}$, $p_{11}$ contain a token, none of $MA5LR$ transitions is enabled. 

Then subnet $MD5LR$ works. The only firable transition $t_{21}$ fires in $L \bdiv 5$ instances putting $L \bdiv 5$ tokens into place $p_{14}$. Then $t_{22}$ fires moving zero from $p_8$ ($MOVE1$) to $p_{12}$. The only firable transition $t_{24}$ fires in $L \bmod 5$ instances putting $L \bmod 5$ tokens into place $X$ ($X:=L \bmod 5$). Then $t_{25}$ fires moving zero from $p_{12}$ to $p_{13}$. The only firable transition $t_{27}$ fires in $L \bdiv 5$ instances returning $L \bdiv 5$ tokens into place $L$. Then $t_{28}$ fires taking a token from place $LEFT$ and putting a token into each place $p_{13}$ and $p_5$ ($STEP$) that disables transitions of subnet $MD5LR$ and enables transitions of subnet $FS$. The result of subnet $MA5LR$ work is the following $X:=L \bmod 5,L:=L \bdiv 5,LEFT:=0$ and marking of internal place $p_{14}$ is equal to zero that is achieved via the only transitions firing sequence ${t_{21}^{L \bdiv 5}}{t_{22}}{t_{24}^{L \bmod 5}}{t_{25}}{t_{27}^{L \bdiv 5}}{t_{28}}$. As the control flow places $p_8$, $p_{12}$, $p_{13}$ contain a token, none of $MA5LR$ transitions is enabled. 

In the same way the proof could be composed for an instruction with the right move when $RIGHT=1,LEFT=0$. The only transitions firing sequence is ${t_{11}^{L}}{t_{13}}{t_{14}^{L \cdot 5}}{t_{16}}{t_{17}^{X}}{t_{19}}{t_{21}^{R \bdiv 5}}{t_{22}}{t_{24}^{R \bmod 5}}{t_{25}}{t_{27}^{R \bdiv 5}}{t_{29}}$. In the both cases the number of transitions equals to 12; the number of their firing insrances, when it is greater than unit, is written as a superscript. Thus, not more than 13 transitions fire. 

\end{proof}

\begin{theorem}
UPN(14,29) simulates WUTM(2,4) in time ${O(14 \cdot k)}$ and space ${O(k)}$, where ${k}$ is the number of WUTM(2,4) steps.
\label{the:polyupn-wutm}
\end{theorem}

Theorem~\ref{the:polyupn-wutm} is an immediate conclusion of lemma~\ref{lem:tmfs}, lemma~\ref{lem:tmtape} and analogous theorem proven in \cite{Zv13-14-29} for UPN(14,29).

\section{Conclusions}

A deterministic arithmetic Petri net PolyUPN(15,29) with 15 places and 29 transitions was constructed, that simulates weakly universal Turing machine of Neary and Woods \cite{Ne08,NW09W} with 2 states and 4 symbols in linear time. Consequently it simulates a given TM in polynomial time.

Moreover, considering the chain of translations described in \cite{Zv13-14-29}, we conclude that PolyUPN(15,29) simulates a given DIPN in polynomial time. Strict complexity estimations require encoding of a given DAPN with gliders of cellular automaton 110 \cite{C04,C08} because WUTM(2,4) simulates its work. In this concern, DAPN could be translated in polynomial time to either of: gliders, cyclic tag system, 2-tag system, Turing machine, or bi-tag system. 

Thus, small universal Petri nets are efficient that justifies their application as models of high performance computations and further development of Petri net based paradigm of computation \cite{Zv12PNPC}.

\section*{Acknowledgement}

The author would like to thank: Erik Winfree whose MCU2013 (http://mcu2013.ini.uzh.ch) invited talk inspired author to remember and apply early studied class of timed Petri nets with multichannel transitions; Turlough Neary, Damien Woods and Kenichi Morita for thier valuable comments on the initial idea.


\nocite{*}
\bibliographystyle{eptcs}

\end{document}